\documentclass[sigconf]{acmart}

\usepackage[T1]{fontenc}
\usepackage[utf8]{inputenc}
\usepackage[english]{babel}
\usepackage{amsmath,amssymb,amsfonts,amsthm}
\usepackage{cmap}
\usepackage{microtype}
\usepackage{mathtools}  %
\usepackage{mathrsfs}
\usepackage[inline,shortlabels]{enumitem}

\usepackage{algorithm}
\usepackage[noend]{algpseudocode}
\floatstyle{ruled}
\newfloat{algo}{tp}{lop}
\floatname{algo}{Algorithm}
\newtheorem{proposition}{Proposition}
\newtheorem{lemma}[proposition]{Lemma}
\newtheorem{theorem}[proposition]{Theorem}

\theoremstyle{acmdefinition}
\newtheorem{definition}[proposition]{Definition}
\newtheorem{remark}[proposition]{Remark}

\allowdisplaybreaks[4]

\newcommand{\mop}[1]{\operatorname{#1}}
\newcommand{\ud}{\mathrm{d}}
\newcommand{\ddt}{\frac{\ud}{\ud t}}

\newcommand{\ddtp}[1]{\frac{\ud^{#1}}{\ud t^{#1}}}
\newcommand{\st}{\ \middle|\ }
\newcommand{\eqdef}{\triangleq}

\newcommand{\bP}{{\mathbb{P}}}
\newcommand{\bN}{{\mathbb{N}}}
\newcommand{\bC}{{\mathbb{C}}}
\newcommand{\bQ}{{\mathbb{Q}}}
\newcommand{\bR}{{\mathbb{R}}}

\def\proj{\mathrm{pr}}
\newcommand{\vol}{\mop{vol}}

\renewcommand\geq{\geqslant}
\renewcommand\leq{\leqslant}
\renewcommand\epsilon\varepsilon

\title{Computing the volume of compact semi-algebraic sets}
\date{\today}

\author{Pierre Lairez}
\affiliation{
\institution{Inria}
\state{France}
}
\email{pierre.lairez@inria.fr}

\author{Marc Mezzarobba}
\affiliation{%
  \institution{Sorbonne Universit\'e, \textsc{CNRS}, \\
    Laboratoire d'Informatique de Paris~6, \textsc{LIP6},
    \'Equipe \textsc{PeQuaN}}
  \city{F-75252, Paris Cedex 05} 
  \country{France}
}
\email{marc@mezzarobba.net}

\author{Mohab {Safey El Din}}
\affiliation{%
  \institution{Sorbonne Universit\'e, \textsc{CNRS}, {Inria},\\
    Laboratoire d'Informatique de Paris~6, \textsc{LIP6},
    \'Equipe \textsc{PolSys}}
  \city{F-75252, Paris Cedex 05} 
  \country{France}
}
\email{mohab.safey@lip6.fr}

\setcopyright{rightsretained}

\copyrightyear{2019}
\acmYear{2019}
\setcopyright{licensedothergov}
\acmConference[ISSAC '19]{International Symposium on Symbolic and Algebraic Computation}{July 15--18, 2019}{Beijing, China}
\acmBooktitle{International Symposium on Symbolic and Algebraic Computation (ISSAC '19), July 15--18, 2019, Beijing, China}
\acmPrice{15.00}
\acmDOI{10.1145/3326229.3326262}
\acmISBN{978-1-4503-6084-5/19/07}

\keywords{ Semi-algebraic sets; Picard-Fuchs equations; Symbolic-numeric
  algorithms } \thanks{%
  Marc Mezzarobba is supported in part by ANR grant ANR-14-CE25-0018-01
  \textsc{FastRelax}. Mohab Safey El Din is supported by the ANR grants
  ANR-17-CE40-0009 \textsc{Galop}, ANR-18-CE33-0011 \textsc{Sesame}, the PGMO
  grant \textsc{Gamma} and the European Union’s Horizon 2020 research and
  innovation programme under the Marie Skłodowska-Curie grant agreement N°
  813211 (POEMA).}

\begin{abstract}
  Let $S\subset \bR^n$ be a compact basic semi-algebraic set defined as the real
  solution set of multivariate polynomial inequalities with rational
  coefficients. We design an algorithm which takes as input a polynomial system
  defining $S$ and an integer $p\geq 0$ and returns the $n$-dimensional volume of $S$ at
  absolute precision $2^{-p}$.

  Our algorithm relies on the relationship between volumes of semi-algebraic
  sets and periods of rational integrals.
  It makes use of algorithms computing the
  Picard-Fuchs differential equation of appropriate periods, properties of critical
  points, and high-precision numerical integration of differential equations.

  The algorithm runs in essentially linear time with respect to~$p$.
  This improves upon the previous exponential bounds obtained by Monte-Carlo
  or moment-based methods.
  Assuming a conjecture of Dimca,
  the arithmetic cost of the algebraic subroutines for computing Picard-Fuchs
  equations and critical points is singly exponential in $n$ and polynomial in
  the maximum degree of the input.
\end{abstract}

\begin{document}

\maketitle

\section{Introduction}\label{sec:intro}

Semi-algebraic sets are the subsets of $\bR^n$ which are finite unions of real
solution sets to polynomial systems of equations and inequa\-lities with
coefficients in $\bR$.  Starting from Tarski's algorithm for
quantifier elimination \cite{Tarski} improved by Collins through the Cylindrical
Algebraic Decomposition algorithm \cite{Collins},
effective real algebraic geometry yields numerous algorithmic
innovations and asymptotically faster routines for problems like deciding
the emptiness of semi-algebraic sets, answering connectivity queries or
computing Betti numbers \cite[e.g.,][]{BaPoRo06,SaSc03,Canny88,SaSc17,BCL18}.
The output of all these algorithms is algebraic in nature.
In this paper, we study the problem of computing the volume of a
(basic) compact semi-algebraic set $S\subset \bR^n$ defined over~$\bQ$.
The output may be transcendental:
for instance, the area of the unit circle in $\bR^2$ is $\pi$.

Volumes of semi-algebraic sets actually lie in a special class of real
numbers, for they are closely related to \emph{Kontsevich-Zagier periods}
introduced in \cite{KontsevichZagier_2001}. A (real) period is the value of an
absolutely convergent
integral of a rational function with rational coefficients over a semi-algebraic
set defined by polynomials with rational coefficients.
For example, algebraic numbers are periods, as are $\pi$, $\log 2$, $\zeta(3)$.
Since $\vol S = \int_S 1 \ud x$, volumes of semi-algebraic sets
defined over~$\bQ$ are periods. Conversely, interpreting an integral as a
``volume under a graph'' shows that periods are differences of volumes of
semi-algebraic sets defined over~$\bQ$. In \cite{Viu-Sos_2015}, it is further shown
that periods are differences of volumes of \emph{compact} semi-algebraic
sets defined over~$\bQ$.

The problem we consider in this paper is thus a basic instance of the more
general problem of integrating an algebraic function over a
semi-algebraic set;
it finds applications in numerous areas of engineering
sciences. Performing these computations at high precision
(hundreds to thousands of digits) is also relevant in experimental
mathematics, especially for discovering formulas, as explained, for example, in
\cite{BaileyBorwein_2011}. Most of the examples featured in this reference are
periods, sometimes in disguise.

\paragraph{Prior work.} The simplest semi-algebraic sets one can consider are
polytopes. The computation of their volume has been extensively studied, with a
focus on the complexity with respect to the dimension. It is known that even
approximating the volume of a polytope deterministically is \#P-hard
\cite{DyerFrieze_1988,Kh89}. The celebrated probabilistic approximation
algorithm in \cite{DyerFriezeKannan_1991}, which applies to more general convex
sets, computes an $\epsilon$-approximation in time polynomial in the dimension
of the set and $1/\varepsilon$. A key ingredient for this algorithm is a Monte
Carlo method for efficiently sampling points from a convex set. Since then,
Monte Carlo schemes have been adopted as the framework of several volume
estimation algorithms.

In contrast, we deal here with compact semi-algebraic sets which can be
non-convex and even non-connected.
Additionally, while volumes of polytopes are rational, the arithmetic nature of
volumes of semi-algebraic sets is much less clear, as unclear as the nature of
periods.
This raises the question of the computational complexity of a volume, even
taken as a single real number.

A simple Monte Carlo technique applies in our setting as well:
one samples points uniformly in a box containing $S$ and estimates the
probability that they lie in $S$. This method is of practical interest at
low precision but requires~$2^{\Omega(p)}$ samples to achieve an
error bounded by $2^{-p}$ with high probability. We refer to \cite{Ko95}
which deals with \emph{definable} sets, a class which encompasses
semi-algebraic sets.

In a different direction, numerical approximation schemes based on the moment
problem and semi-definite programming have been designed in \cite{HLS09}.
They are also of practical interest at low precision, and can provide
rigorous error bounds, but the convergence is worse than exponential
with respect to~$p$~\cite{KordaHenrion_2017}.

Another line of research, going back to the nineteenth century, is concerned with
the computation of periods of algebraic varieties.
In particular, we build on work \cite{ChudnovskyChudnovsky_1990} on the
high-precision numerical solution of ODEs with polynomial coefficients which was
motivated, among other things, by applications to periods of Abelian
integrals~\cite[see][p.~133]{ChudnovskyChudnovsky_1990}.

\paragraph*{Main result.} We describe a new strategy for computing volumes
of semi-algebraic sets, at the crossroads of effective algebraic and real
algebraic geometry, symbolic integration, and rigorous numerical computing. Our
approach effectively reduces the volume computation to the setting
of~\cite{ChudnovskyChudnovsky_1990}. It yields an algorithm that approximates
the volume of a {\em fixed}, bounded basic semi-algebraic set in
almost linear time with respect to the precision.
More precisely, we prove the following bit complexity estimate.

\begin{theorem}\label{thm:main}
  Let $f_1, \ldots, f_r$ be polynomials in $\bQ[x_1, \ldots, x_n]$, and let
  $S\subset \bR^n$ be the semi-algebraic set defined by $f_1\geq 0, \ldots,
  f_r\geq 0$. Assume that $S$ is compact. There exists an algorithm which
  computes, on input~$p \geq 0$ and $(f_1, \ldots, f_r)$, an
  approximation~$\tilde V$ of the volume~$V$ of $S$ with $|\tilde V - V| \leq
  2^{-p}$.
  When $f_1, \dots, f_r$ are fixed, the algorithm runs in time
  $O(p\log(p)^{3+\varepsilon})$ (for any $\varepsilon > 0$) as $p \to \infty$.
\end{theorem}

The algorithm recursively computes integrals of volumes of sections of~$S$.
Let~$v(t)$ denote the $(n-1)$-dimensional volume of~$S\cap
\proj^{-1}(t)$, for some nonzero linear projection~$\proj : \bR^n\to \bR$.
In our setting, $v$~is a piecewise analytic function and, except at finitely
many~$t$, is solution of a linear differential equation with polynomial
coefficients known as a Picard-Fuchs equation.

The problem points belong to the critical locus of the restriction of the
projection $\proj$ to a certain hypersurface containing the boundary of $S$ and are
found by solving appropriate polynomial systems.
(Compare~\cite{Kh93} in the case of polytopes.)
The Picard-Fuchs equation for $v$ is
produced by algorithms from symbolic integration, in particular
\cite{BostanLairezSalvy_2013, Lairez_2016}. To obtain the volume of~$S$, it
then suffices to compute~$\int_\bR v$ with a rigorous numerical ODE solving
algorithm, starting from values~$v(\rho_i)$ at suitable points~$\rho_i$ obtained
through recursive calls.

The complexity with respect to the dimension~$n$ of the ambient space and the
number~$r$, maximum degree~$D$, and coefficient size of the polynomials $f_i$ is
harder to analyze.
We will see, though, that under reasonable assumptions, the ``algebraic'' steps
(computing the critical loci and of the Picard-Fuchs equations) take at most
$(rD)^{O(n^2)}$ arithmetic operations in $\bQ$.

\paragraph{Example.} The idea of the method is well illustrated by the
example of a torus $S$, here of major radius~$2$ and minor radius~$1$. Let
\[ S = \left\{ (x,y,z)\in \bR^3 \st (x^2+y^2+z^2 + 3)^2 \leq 16 (x^2+y^2) \right\}. \]
The area (2-dimensional volume) of a section~$S \cap \left\{ x=t \right\}$
defines a function~$v : \bR \to \bR$ (see Figure~\ref{fig:volume}).
\begin{figure}
  \centering
  \includegraphics{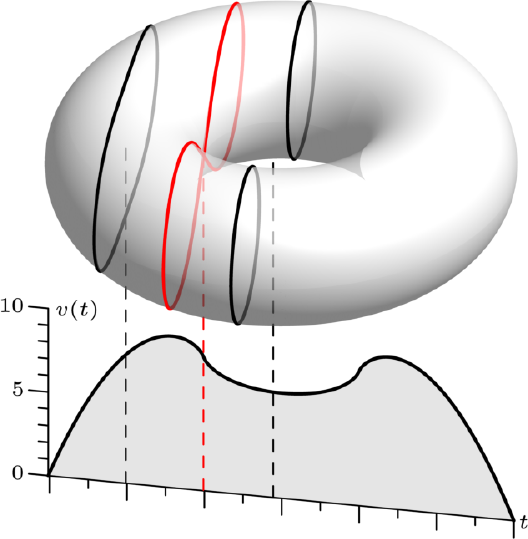}
  \vspace*{-.5\baselineskip}
  \caption{Volume of the sections of the torus $S$ as a function of the
    parameter~$t$. In red, a singular section.}
  \label{fig:volume}
\end{figure}
It is analytic, except maybe at the critical values~$t=\pm 3$ and~$t=\pm 1$
where the real locus of the curve $(t^2+y^2+z^2 + 3)^2 = 16 (t^2+y^2)$ is
singular.
On each interval on which~$v$ is analytic, it satisfies the Picard-Fuchs
equation
\begin{multline}\label{eq:2}
  (t-3)(t+3)(t-1)^2(t+1)^2t^2 v'''(t)+(t^2+9)(t-1)^2(t+1)^2t
  v''(t)  \\
  -(2t^4+11t^2-9)(t-1)(t+1) v'(t)+2(t^2+3)t^3 v(t) = 0,
\end{multline}
which we compute in 2~seconds on a laptop using the algorithm
of~\cite{Lairez_2016} and Theorem~\ref{thm:volume-section2}.

We know some special values of~$v$, namely~$v(0) = 2\pi$, $v(\pm1) = 8$
and~$v(\pm 3) = 0$.
Additionally, we have~$v(3 \pm t) = O(t^2)$ as~$t\to \mp 0$.
These properties characterize the analytic function~$v_{|(-1,1)}$ in the
$2$-dimensional space of analytic solutions of the differential
equation~\eqref{eq:2} on~$(-1,1)$, and similarly for $v_{|(1,3)}$.
(Our algorithm actually uses recursive calls at generic points instead of these
\emph{ad hoc} conditions.)
The rigorous ODE solver part of the Sage package ore\_algebra
\cite{Mezzarobba_2016} determines in less than a second that
\[ \int_{-3}^3 v(t)\ud t = 39.4784176043[...]25056533975 \pm 10^{-60}. \]
And indeed, it is not hard to see in this case that~$\vol S = 4 \pi^2$.
We can obtain 1000 digits in less than a minute.

\paragraph{Outline.}

The remainder of this article is organized as follows.
In Section~\ref{sec:volume-semi-sets}, we give a high-level description of the
main algorithm.
As sketched above, the algorithm relies on the computation of critical points,
Picard-Fuchs equations, and numerical solutions of these equations.
In Section~\ref{sec:peri-depend-param}, we discuss the computation of
Picard-Fuchs equations and critical points, relating these objects with
analyticity properties of the ``section volume'' function.
Then, in Section~\ref{sec:numer-comp-with}, we describe the numerical solution
process and study its complexity with respect to the precision.
Finally, in Section~\ref{sec:complexity}, we conclude the proof of
Theorem~\ref{thm:main} and state partial results on the complexity of the
algorithm with respect to $n$, $r$, and~$D$.

\paragraph{Acknowledgements.}

We would like to thank the anonymous reviewers for their careful reading and
valuable comments.

\section{Volumes of semi-algebraic sets}
\label{sec:volume-semi-sets}

\def\Single{A}

We start by designing an algorithm which deals with the case of a union of
connected components of a semi-algebraic set defined by a single inequality.
Next, we will use a deformation technique to handle semi-algebraic sets defined
by several inequalities.

\subsection{Sets defined by a single inequality}
\label{sec:volume-section}

Let~$f \in \bQ[t,x_1,\dotsc,x_n]$ and $\Single$ be the semi-algebraic set
\[ \Single \eqdef \left\{ (\rho,x) \in \bR\times\bR^{n} \st f(\rho,x) \geq 0 \right\}. \]
Let~$\proj : \bR^{n+1} \to \bR$ be the projection on the $t$-coordinate.
We want to compute the volume of a union~$U$ of connected components
of~$A$ starting from the volumes of suitable fibers $U \cap \proj^{-1}(\rho)$.
For technical reasons, we first consider the slightly more general
situation where~$U$ is a union of connected components of $\Single \cap
\proj^{-1}(I)$ for some open interval~$I \subseteq \bR$.
From a computational point of view, we assume that~$U$ is described by a
semi-algebraic formula~$\Theta_U$, that is,
\[ U = \left\{ (\rho,x) \in \Single \st \Theta_U(\rho,x) \right\}, \]
where $\Theta_U$ is a finite disjunction of conjunctions of
polynomial inequalities with (in our setting) rational coefficients.

For $\rho\in I$, let~$U_\rho \eqdef U \cap \proj^{-1}(\rho)$ and
$v(\rho) \eqdef \vol_n U_\rho$.
Let~$\Sigma_f \subseteq \bR$ (we will often omit the subscript~$f$) be the set
of \emph{exceptional values}
\begin{equation}
  \label{eq:Sigma}
\Sigma_f \eqdef 
  \left\{ \rho \in \bR \mid \exists x\in\bR^n, f(\rho,x) = 0
  \wedge \forall i, 
  \tfrac{\partial}{\partial x_i} f(\rho, x) = 0 \right\}.  
\end{equation}
Thus, when $f$ is square-free, exceptional values are either critical values of
the restriction to the hypersurface $\{f=0\}$ of the projection~$\proj$, or
images of singular points of $\{f = 0\}$.
By definition of $\Sigma$, for any~$\rho\in\bR \setminus \Sigma$, the zero set
of $f_\rho = f(\rho,-)$ is a smooth submanifold of~$\bR^n$. 

Further, we say that assumption \eqref{eq:R} holds for $f$ if
\begin{equation} \label{eq:R}\tag{R}
  \left\{ z\in\mathbb{R}^{n+1} \st
  f(z)=0
  \wedge \tfrac{\partial}{\partial t}f(z)=0
  \wedge \forall i, 
  \tfrac{\partial}{\partial x_i}f(z) = 0
  \right\} = \varnothing.
\end{equation}
Observe that
by Sard's theorem \cite[e.g.][Theorem 5.56]{BaPoRo06},
when \eqref{eq:R} holds,
the exceptional set~$\Sigma$ is finite.

The mainstay of the method is the next result, to be proved
in~\S \ref{sec:peri-depend-param}.
Let $\mathscr D \subset \bQ[t][\ddt]$ denote the set of Fuchsian linear
differential operators with coefficients in~$\bQ[t]$ whose local exponents at
singular points are rational
(see §\ref{sec:numer-comp-with} for reminders on Fuchsian operators and their
exponents).

\begin{theorem}\label{thm:volume-section}
  If $U$\! is bounded and
  $I\cap \Sigma = \varnothing$, then the function~$v_{|I}$
  is solution of a computable differential equation of the form $P(v)=0$,
  where $P \in \mathscr D$ depends only on $f$.
\end{theorem}

We will also use the following proposition, which summarizes the results of
Proposition~\ref{prop:dsolve} and Lemma~\ref{lem:evaluation-points}
in~§\ref{sec:numer-comp-with}.
The complete definition of ``good initial conditions'' is given there as well.
Up to technical details, this simply means a system~$\mathscr I$ of linear
equations of the form $y^{(k)}(u) = s$ that suffices to characterize a
particular solution~$y$ among the solutions of $P(y)=0$.
An \emph{$\epsilon$\nobreakdash-approximation} of~$\mathscr I$ is made of the
same equations with each right-hand side~$s$ replaced by an enclosure~$\tilde s
\ni s$ of diameter $\leq \epsilon$.

\begin{proposition} \label{prop:spec-dsolve}
  Let $P \in \mathscr D$ have order~$m$, and let $J = (\alpha, \beta)$ be a real
  interval with algebraic endpoints.
  Let $y: J \to \bR$ be a solution of $P(y) = 0$ with a finite limit at~$\alpha$
  and $\mathscr I$ be a system of good initial conditions for~$P$ on~$J$
  defining~$y$.

  \begin{enumerate}

  \item \label{item:dsolve}
  Given $P, \alpha$, a precision~$p \in \bN$ and a $2^{-p}$-approximation
  $\tilde{\mathscr I}$ of~$\mathscr I$, one can compute an interval of width%
  ~$O(2^{-p})$ (as $p \to \infty$ for fixed $P$, $\alpha$, and~$\mathscr I$)
  containing $\lim_{t \to \alpha} y(t)$.

  \item \label{item:pickgoodpoints}
  Given $P, \alpha, \beta$, one can compute $\rho_1, \dots, \rho_{m} \in J \cap \bQ$
  such that the $y(\rho_j)$ form a system of good initial
  conditions for~$P$ on~$J$.

  \end{enumerate}
\end{proposition}

Assume now that $U$ is a bounded union of connected components of~$A$ (i.e., that
we can take $I=\bR$ above), and that \eqref{eq:R} holds for $f$.
The algorithm is recursive.
Starting with input $f, \Theta_U$, and~$p$,
it first computes the set $\Sigma =
\{\alpha_1\leq \cdots \leq \alpha_\ell\}$ of exceptional values so as to
decompose $\bR - \Sigma$ into intervals over which the function~$v$ satisfies the
differential equation $P(y) = 0$ given by Theorem~\ref{thm:volume-section}.
Since $U$~is bounded, one has
\[
  \vol_{n+1}U
  = \sum_{i=1}^{\ell-1} \vol_{n+1}
    \left (U \cap \proj^{-1}(\alpha_i, \alpha_{i+1})\right )
  = \sum_{i=1}^{\ell-1} \int_{\alpha_i}^{\alpha_{i+1}} v(t) \, \mathrm d t.
\]

Fix~$i$ and consider the interval $J = (\alpha_i, \alpha_{i+1})$.
Since $v_{|J}$ is annihilated by~$P$,
its anti-derivative $w: J \to \bR$ vanishing at~$\alpha_{i+1}$ is annihilated
by the operator $P \, \ddt$, which belongs to~$\mathscr D$ since $P$~does.
Additionally, if $[ v(\rho_j) = s_j ]_j$ is a system of good initial
conditions for~$P$ that defines~$v_{|J}$, then
$[ w'(\rho_j) = s_j ]_j \cup [ w(\alpha_{i+1}) = 0 ]$
is a system of good initial conditions for $P \, \ddt$ defining~$w$
(see Lemma~\ref{lemma:ini-integral} in~§\ref{sec:numer-comp-with}).
Thus, by Proposition~\ref{prop:spec-dsolve}, to compute $w(\alpha_i)$ to
absolute precision~$p$, it suffices to compute $v(\rho_j)$, $1\leq j \leq m$,
to precision $p + O(1)$.

By definition of $\Sigma$, since $\rho_j \notin \Sigma$,
there is no solution to the system
\[
f(\rho_j,-)=\tfrac{\partial }{\partial x_1}f(\rho_j,-)=\cdots=\tfrac{\partial
  }{\partial x_n}f(\rho_j,-)=0
\]
which means that \eqref{eq:R} holds for $f(\rho_j,-)$.
Additionally, $U\cap \proj^{-1}(\rho_j)$ is a bounded union of connected
components of $A \cap \proj^{-1}(\rho_j)$.
Hence, the values $v(\rho_j)$ can be obtained by recursive calls to the
algorithm with $t$ instantiated to~$\rho_j$.

The process terminates since each recursive call handles one less variable.
In the base case, we are left with the problem of computing the length of a
union of real intervals encoded by a semi-algebraic formula.
This is classically done using basic univariate polynomial arithmetic and real
root isolation \cite[Chap. 10]{BaPoRo06}.

\begin{algorithm}[tb]
    \caption{Volume of $U$ at precision $O(2^{-p})$}
    \label{algo:volumerec}
    \begin{algorithmic}[1] %
      \Procedure{Volume1}{$f, \Theta_U, (t, x_1, \ldots, x_n), p$}
      \If {$n=0$}
      return $\textsc{UnivariateVolume}(f, \Theta_U, p)$.
      \EndIf
      \State $(\alpha_1, \ldots, \alpha_\ell) \gets \textsc{CriticalValues}(f, t)$
      \State $P \gets \textsc{PicardFuchs}(f, t)$
      \For{$1 \leq i \leq \ell-1$}
      \Comment{$\tilde s_j$, $\tilde S_i$ are intervals}
      \State $(\rho_1, \ldots, \rho_{m}) \gets \textsc{PickGoodPoints}(P, \alpha_i, \alpha_{i+1})$
      \For{$1\leq j \leq m$}
      \State $\tilde s_j \gets \textsc{Volume1}(f_{|t=\rho_j}, {\Theta_U}_{|t=\rho_j}, (x_1, \ldots, x_n), p)$
      \EndFor
      \State $\tilde{\mathscr I} \gets [
        y'(\rho_1) = \tilde s_1, \dots, y'(\rho_{m}) = \tilde s_{m},
        y(\alpha_{i+1}) = 0 ]$
      \State $\tilde S_i \gets - \textsc{DSolve}(P \ddt, \tilde{\mathscr I}, \alpha_i, p)$
      \EndFor
      \State \textbf{return} $\tilde S_1+\cdots + \tilde S_\ell$
      \EndProcedure
    \end{algorithmic}
\end{algorithm}

The complete procedure is formalized in Algorithm~\ref{algo:volumerec}.
The quantities denoted with a tilde in the pseudo-code are understood to be
represented by intervals, and the operations involving them
follow the semantics of interval arithmetic.
Additionally, we assume that we have at our disposal the following subroutines:
\begin{itemize}
\item $\textsc{PicardFuchs}(f, t)$,
  $\textsc{DSolve}(P, \tilde{\mathscr I}, \alpha, p)$,
  and
  $\textsc{PickGoodPoi}\linebreak[0]\textsc{nts}(P, \alpha, \beta)$,
  which implement the algorithms implied, respectively,
  by Theorem~\ref{thm:volume-section} and
  Proposition~\ref{prop:spec-dsolve}
  (\ref{item:dsolve})~and~(\ref{item:pickgoodpoints});
\item $\textsc{CriticalValues}(f, t)$, which returns an encoding for a
  finite set of real algebraic numbers containing the exceptional values
  associated to~$f$, sorted in increasing order;
\item $\textsc{UnivariateVolume}(g, \Theta_U, p)$ where
  $g \in \bQ[t]$ and
  $\Theta_U$ is a semi-algebraic formula describing a union~$U$ of connected
  components of $\{g\geq 0\}$,
  which returns an interval of width $\leq 2^{-p}$ containing $\vol_1 U$.
\end{itemize}

The following result summarizes the above discussion.

\begin{theorem} \label{thm:volumerec}
  Assume that $U$~is a bounded union of connected components of~$A$
  and that \eqref{eq:R}~holds.
  Then, on input $f, \Theta_U, p$ and $(t, x_1, \ldots, x_n)$,
  Algorithm~\ref{algo:volumerec} ($\textsc{Volume1}$) returns a real
  interval of width $O(2^{-p})$ (for fixed~$f$) containing~$\vol_{n+1} U$.
\end{theorem}

\subsection{Sets defined by several inequalities}
\label{sec:general-semi-algebraic}

Now, we show how to compute the volume of a basic
semi-algebraic set $S\subset \bR^n$ defined by
\[
f_1\geq 0, \ldots, f_r\geq 0, \qquad f_i \in \bQ[x_1, \dots, x_n],
\]
assuming that $S$ is compact.

We set $f = f_1\cdots f_r - t \in \bQ[t, x_1, \dots, x_n]$, and consider the
semi-algebraic set $A\subset \bR^{n+1}$ defined by $f \geq 0$. Observe that the
polynomial~$f$ satisfies \eqref{eq:R} because~$\frac{\partial f}{\partial t} =
-1$. We can hence choose an interval $I = (0, \alpha)$ with $\alpha \in \bQ$
that contains no element of~$\Sigma_f$. Let $U\eqdef A\cap (I \times S)$ and
$\proj$ be the projection on the $t$-coordinate. For fixed~$\rho \in I$, the set
$U\cap \proj^{-1}(\rho)$ can be viewed as a bounded subset of~$S$, whose volume
$v(\rho) = \vol_n(U\cap \proj^{-1}(\rho))$ tends to $\vol_n S$ as $\rho \to 0$.

The set
$U$ itself is bounded and the formula
\[ \Theta_U = f_1 \geq 0\wedge \dotsb \wedge f_r \geq 0\wedge 0 < t < \alpha \]
defines $U$ in $A$.
In addition, $U$ is a union of connected components of $A \cap \proj^{-1}(I)$.
Indeed, for any point~$(\rho,x) \in A$ with $\rho \in I$, it
holds that $f_1(x)\dotsb \allowbreak f_r(x) > 0$.
This implies that $U = A \cap (I \times \mathring S)$ where
$\mathring S$ is the interior of~$S$.
Therefore, $U$ is both relatively closed (as the trace of $\bR \times S$) and
open (as that of $\bR \times \mathring S$) in $A \cap \proj^{-1}(I)$.

We are hence in the setting of the previous subsection.
Since $I\cap \Sigma_f = \varnothing$ by definition of~$I$,
Theorem~\ref{thm:volume-section} applies,
and the function $v: I \to \bR$ is annihilated by an operator $P \in \mathscr D$
which is computed using the routine $\textsc{PicardFuchs}$ introduced
earlier.
By Proposition~\ref{prop:spec-dsolve}, one can choose rational points~$\rho_j
\in I$ such that the values of~$v$ at these points characterize it among the
solutions of~$P$, and, given sufficiently precise approximations of~$v(\rho_j)$,
one can compute $\vol_n S = \lim_{t\to0} v(t)$ to any desired accuracy.

The ``initial conditions'' $v(\rho_j)$ are computed by calls to
Algorithm~\ref{algo:volumerec} with $f$~and~$\Theta_U$ specialized to
$t=\rho_j$.
In the notation of §\ref{sec:volume-section}, this corresponds to taking
$A = A(\rho_j) = \{ f_1 \cdots f_r \geq \rho_j \}$
and
$U = U(\rho_j) = A(\rho_j) \cap S$.
Thus, $U(\rho_j)$ is compact, and, since no~$f_i$ can change sign on a connected
component of $A(\rho)$ for $\rho>0$, it is the union of those
connected components of~$A(\rho_j)$ where $f_1, \dots, f_r \geq 0$.
Additionally, \eqref{eq:R} holds for $f(\rho_j, -)$ since $\rho_j \notin \Sigma_f$.
Therefore, the assumptions of Theorem~\ref{thm:volumerec} are satisfied.

\begin{algorithm}[t]
    \caption{Volume of $S$}
    \label{algo:volume}
    \begin{algorithmic}[1] %
      \Procedure{Volume}{$(f_1, \ldots, f_r), p$}
      \State $f\gets f_1\cdots f_r -t$
      \State $(\alpha_1, \ldots, \alpha_\ell) \gets \textsc{CriticalValues}(f,
      t)$
      \State $\alpha \gets \text{a rational s.t.\ $0 < \alpha < \min (\{\alpha_i \mid \alpha_i >0\}\cup\{1\})$}$
      \State $\Theta_U \gets f_1 \geq 0\wedge \dotsb \wedge f_r \geq 0\wedge 0 < t < \alpha$
      \State $P \gets \textsc{PicardFuchs}(f, t)$
      \State $(\rho_1, \ldots, \rho_{m}) \gets \textsc{PickGoodPoints}(P, 0, \alpha)$
      \For{$1\leq j \leq m$}
      \Comment{$\tilde s_j$ are intervals}
      \State $\tilde s_j \gets \textsc{Volume1}({f}_{|t=\rho_j}, (\Theta_U)_{|t=\rho_j}, (x_1, \ldots, x_n), p)$
      \EndFor
      \State \textbf{return} $\textsc{DSolve}(P, [ y(\rho_j) = \tilde s_j ]_{j=1}^{m}, p)$
      \EndProcedure
    \end{algorithmic}
\end{algorithm}

We obtain Algorithm~\ref{algo:volume} (which uses the same subroutines and
conventions as Algorithm~\ref{algo:volumerec}) and the following correctness
theorem.

\begin{theorem}\label{thm:general}
  Let~$f_1,\dotsc,f_r \in \bQ[x_1,\dotsc,x_n]$.
  Let $S$ be the semi-algebraic set defined by $f_1\geq 0, \dotsc, f_r\geq 0$.
  Assume that $S$~is bounded.
  Then, given $(f_1, \ldots, f_r)$ and a working precision $p \in \bN$,
  Algorithm \ref{algo:volume} ($\textsc{Volume}$)
  computes an interval containing $\vol_n(S)$
  of width $O(2^{-p})$ as $p \to \infty$ for fixed $f_1, \dots, f_r$
\end{theorem}

\begin{remark}
  In case~$S$ has empty interior, Algorithm~\ref{algo:volume}
  returns zero.
  When $S$ is contained in a linear subspace of dimension $k <n$, one could in
  principle obtain the $k$-volume of~$S$ by computing linear equations defining
  the subspace (using quantifier elimination as
  in~\cite{KhachiyanPorkolab2000,SafeyElDinZhi2010}) and eliminating $n-k$
  variables.
  The new system would in general have algebraic
  instead of rational coefficients, though.
\end{remark}

Lastly, we note that a more direct symbolic computation of integrals on general semi-algebraic sets depending
on a parameter is possible with Oaku's algorithm~\cite{Oaku_2013}, based on the
effective theory of $\mathcal{D}$-modules. 

\section{Periods depending on a parameter}
\label{sec:peri-depend-param}

Let us now discuss in more detail the main black boxes used by the volume
computation algorithm.
In this section, we study how the volume of a section $U \cap \proj^{-1}(\rho)$
varies with the parameter~$t = \rho$.

\subsection{Picard-Fuchs equations}

Let~$R(t,x_1,\dotsc,x_n)$ be a rational function. A \emph{period of the
  parameter-dependent rational
  integral $\oint R(t,x_1,\dotsc,x_n)\,\ud x_1\dotsb\ud x_n$} is an analytic
function $\phi : \Omega\to \bC$, for some open subset $\Omega$ of~$\bR$ or~$\bC$ such
that for any~$s\in \Omega$ there is an $n$-cycle $\gamma\subset \bC^n$ and a neighborhood~$\Omega' \subset \Omega$ of~$s$ such that
for any~$t\in \Omega'$, $\gamma$ is disjoint from the poles of~$R(t,-)$ and
\begin{equation} \label{eq:period}
  \phi(t) = \int_\gamma R(t,x_1,\dotsc,x_n) \,\ud x_1 \dotsb \ud x_n.
\end{equation}
Recall that an $n$-cycle is a compact $n$-dimensional real submanifold of
$\bC^n$ and that such an
integral is invariant under a continuous deformation of the integration
domain~$\gamma$ as long as it stays away from the poles of~$R(t,-)$, as a
consequence of Stokes' theorem.
It is also well known that such a function~$\phi$ depends analytically on~$t$,
by Morera's theorem for example.

For instance, algebraic functions are periods:
if~$\phi : \Omega\to \bC$ satisfies a nontrivial
relation~$P(t,\phi(t))=0$, with square-free~$P\in\bC[t,x]$,
then $\phi(t)$ is a period by the residue theorem applied to
\[ \phi(t) = \frac{1}{2\pi i}\oint_{\gamma} \frac{x}{P(t,x)} \frac{\partial P}{\partial
  x}(t,x) \,\ud x \]
where $\gamma \subset \bC$ encloses~$\phi(t)$ and no other root of~$P$.
Indeed, the integrand decomposes as
$ \sum_{i=1}^{\deg_x P} x/(x-\psi_i(t))$,
where the functions~$\psi_i$ parametrize the roots of~$P(t,-)$, and, w.l.o.g.,
$\phi = \psi_1$.

Periods of rational functions are solutions of Fuchsian linear differential equations
with polynomial coefficients known as Picard-Fuchs equations. This was proved
in \cite{Picard_1902} in the case of three variables at most and a parameter and
generalized later, using either the finiteness of the algebraic De Rham
cohomology \cite[e.g.][]{Grothendieck_1966a,Monsky_1972,Christol_1985}
or the theory of D-finite functions \cite{Lipshitz_1988}.
The regularity of Picard-Fuchs equations is due to Griffiths \cite[see][]{Katz_1971}.

\begin{theorem}\label{thm:picard-fuchs}
  If~$\phi : \Omega\to \bC$ is the period of a rational integral then~$\phi$ is
  solution of a nontrivial linear differential equation with polynomial
  coefficients $P(\phi) = 0$, where the operator~$P$ belongs to the
  class~$\mathscr D$ introduced in~§\ref{sec:volume-section}.
\end{theorem}

Several algorithms are known and implemented to compute such Picard-Fuchs
equations \cite{Lairez_2016,Koutschan_2010,Chyzak_2000}.

\begin{theorem}[\cite{BostanLairezSalvy_2013}]
  \label{thm:algo-periods}
  A period of the form~\eqref{eq:period}
  is solution of a differential equation of order at most~$D^n$
   where~$D$ is the degree of~$R$; and one
   can compute such an equation in 
   $D^{O(n)}$ operations in~$\bQ$.
\end{theorem}

Note however that the algorithm underlying this result might not return the
equation of minimal order, but rather a left multiple of \emph{the} Picard-Fuchs
equation. So there is no guarantee that the computed operator belongs
to~$\mathscr D$. On the other hand, Lairez's algorithm~\cite{Lairez_2016}
can compute a sequence of operators with non-increasing order which eventually
stabilizes to the minimal order operator. In particular, as long as the computed
operator is not in~$\mathscr D$, we can compute the next one, with the guarantee
that this procedure terminates. A conjecture of Dimca~\cite{Dimca_1991} ensures
that it terminates after at most~$n$ steps, leading to a~$D^{O(n)}$ complexity
bound as in Theorem~\ref{thm:algo-periods}.

\subsection{Volume of a section and proof of Theorem~\ref{thm:volume-section}}

We prove Theorem~\ref{thm:volume-section} as a consequence of
Theorem~\ref{thm:picard-fuchs} and the following result. It is probably well
known to experts but it is still worth an explicit proof.
We use the notation %
of \S \ref{sec:volume-semi-sets}.

\begin{theorem}\label{thm:volume-section2}
  If\/ $I \cap \Sigma = \varnothing$ and if\/ $U$ is bounded then
  the function~$\rho\in I \mapsto \vol_n U_\rho$ is a period of the rational
  integral
  \[ \frac{1}{2i\pi}\oint \frac{x_1}{f_\rho} \frac{\partial f_\rho}{\partial x_1}
  \ud x_1 \dotsb \ud x_n. \]
\end{theorem}

\begin{proof}
  Let~$\rho\in I$.
  By Stokes' formula,
  \begin{equation*}\textstyle
    \vol_n U_\rho = \int_{U_\rho} \ud x_1 \dotsb \ud x_{n} = \oint_{\partial U_\rho} x_1 \ud x_2 \dotsb \ud x_{n},
  \end{equation*}
  where~$\partial U_\rho$ is the boundary of~$U_\rho$.
  Due to the regularity
  assumption $\rho\not\in \Sigma$, the gradient $\nabla_p f_\rho$ does not vanish on
  the real zero locus of~$f_\rho$, denoted~$V(f_\rho)$.
  Because $U_\rho$ is a union of
  connected components of~$A\cap \proj^{-1}(\rho)$,
  it follows that $\partial U_\rho$
  is a compact $(n-1)$-dimensional submanifold of~$\bR^{n+1}$ contained in~$V(f_\rho)$.

  For $\epsilon > 0$, let $\tau(\rho)$ be the \emph{Leray tube} defined by
  \[ \tau(\rho) \eqdef \left\{ p + u \nabla_p f_\rho \st p \in \partial U_\rho, u \in \bC
      \text{ and } |u| = \epsilon \right\}. \]
  This is an $n$-dimensional submanifold of~$\bC^n$.
  We choose $\epsilon$ small enough that~$\tau(\rho) \cap V(f_\rho) = \varnothing$:
  this is possible because $\nabla_p f$ does not vanish on~$\partial U_\rho$ which
  is compact.

  Let $R(\rho, x_1,\dotsc,x_n) = x_1 f_\rho^{-1} {\partial f_\rho}/{\partial x_1}$;
  observe that $\tau(\rho)$ does not cancel the denominator of $R(\rho,
  x_1,\dotsc,x_n)$. Leray's residue theorem~\cite{Leray_1959}  shows that
  \begin{align*}\textstyle
    2\pi i \oint_{\partial U_t} x_1 \ud x_2 \dotsb \ud x_{n}
     = & \textstyle\oint_{\tau(\rho)} \frac{\ud f_\rho}{f_\rho}
                          \wedge (x_1 \ud x_2\dotsb\ud x_{n}) \\
    = & \textstyle \oint_{\tau(\rho)} R(\rho, x_1,\dotsc,x_n) \, \ud x_1 \dotsb \ud x_n.
  \end{align*}
  (In Pham's \cite[Thm.~III.2.4]{Pham_2011} notation, we have $\gamma = \partial U_\rho$,
    $\delta\gamma = \tau(\rho)$, $\varphi =  f_ \rho^{-1} \ud f_\rho
    \wedge (x_1 \ud x_2\dotsb\ud x_{n})$,
    and $\operatorname{res}[\varphi] = x_1 \ud x_2\dotsb\ud x_{n}$.)
 
  To match the definition of a period and conclude the proof, it is enough to
  prove that, locally, the integration domain~$\tau(\rho)$ can be made independent
  of~$\rho$.
  And indeed, since $U$~is a union of connected components of
  $A \cap \proj^{-1}(I)$, we have $\partial U \subseteq {f=0}$.
  Therefore, since $I$~is connected and $I \cap \Sigma = \varnothing$, the
  restriction of the projection~$\proj$ defines a submersive map from
  $\partial U \cap \proj^{-1}(I)$ onto~$I$.
  Additionally, $\partial U$ is compact, hence this map is proper.
  Ehresmann's theorem then implies that there exists a continuous map
  $h : I \times \partial U_\rho \to \bR^n$ such that~$h(\sigma,-)$ induces a
  homeomorphism~$\partial U_\rho \simeq \partial U_\sigma$ for any~$\sigma\in I$.
  In particular, we have
  \[ \tau(\sigma) = \left\{ h(\sigma,p) + u \nabla_{h(\sigma,p)} f_\sigma \st  p \in \partial U_\sigma, u \in \bC
      \text{ and } |u| = \epsilon \right\}. \]
  This formulation makes it clear that~$\tau(\sigma)$ deforms continuously
  into~$\tau(\rho)$ as $\sigma$~varies.
  Since $\tau(\sigma)$ does not intersect the polar locus $V(f_\sigma)$ of $R(\sigma, -)$,
  neither does $\tau(\rho)$ when $\sigma$~and~$\rho$ are close enough,
  by compactness of~$\tau(\sigma)$ and continuity of the deformation.
  Therefore, given any $\rho \in I$, we have
  $\oint_{\tau(\sigma)} R(s, -) = \oint_{\tau(\rho)} R(\sigma, -)$
  for $\sigma$ close enough to~$\rho$.
\end{proof}

The choice of $x_1 \ud x_2 \dots \ud x_n$ as a primitive of
$\ud x_1 \dots \ud x_n$ in Theorem~\ref{thm:volume-section2}
is arbitrary, but of little consequence, since the Picard-Fuchs equation only
depends on the cohomology class of the integrand.

\subsection{Critical values}

Theorem~\ref{thm:volume-section} does not guarantee that~$v$ satisfies the
Picard-Fuchs equation on the whole domain where the equation is nonsingular.
It could happen that the solutions extend analytically across an exceptional
point, or that some of them have singularities between two consecutive
exceptional points.
As a consequence, we need to explicitly compute~$\Sigma$.

\begin{lemma}\label{lemma:sard-algo}
  There exists an algorithm which, given on input a polynomial
  $f\in \bQ[t, x_1, \ldots, x_n]$ of degree $D$
  satisfying~\eqref{eq:R},
  computes a polynomial $g\in \bQ[t]-\{0\}$ of degree $D^{O(n)}$ whose set of
  real roots contains $\Sigma$, using $D^{O(n)}$ operations in $\bQ$.
\end{lemma}

\begin{proof}
  Recall that, when \eqref{eq:R} holds, the set~$\Sigma$ is finite. Our goal is
  to write $\Sigma$ as the root set of a univariate polynomial~$g$. Consider the
  polynomial
  $ h = f^2+ ({\partial f}/{\partial x_1})^2+ \cdots + ({\partial f}/{\partial x_n})^2. $
  We start by computing at
  least one point in each connected component of the real algebraic set defined
  by $h = 0$ using \cite[Algorithm 13.3]{BaPoRo06}. By \cite[Theorem
  13.22]{BaPoRo06}, this algorithm uses $D^{O(n)}$ operations. It returns a
  rational parametrization: polynomials~$P$, $F$, $G_1,\dotsc, G_n$ in
  $\bQ[y]$ of degree $\leq D^{O(n)}$ such that $P$ is square-free and the set of
  points
  \[ \left\{
    {P'(\xi)^{-1}}\bigl(F(\xi), G_1(\xi), \dotsc, G_n(\xi)\bigr) \in \bR^{n+1}
    \st \xi \in \bR, P(\xi) = 0
    \right\} \]
  meets every connected component of the zero set of~$h$. In particular, $\Sigma
  = \left\{ F(\xi) / P'(\xi) \st \xi \in \bR, P(\xi) = 0 \right\}$. As a polynomial~$g$, we take
  the resultant with respect to~$y$ of~$P(y)$ and~$F(y)-t P'(y)$: its set of roots
  contains~$\Sigma$. Since $P$ and $F$ have
  degree $D^{O(n)}$, this last step also uses $D^{O(n)}$ operations
  in $\bQ$ \cite{GathenGerhard_1999}.
\end{proof}

\section{Numerics}
\label{sec:numer-comp-with}

Let us turn to the numerical part of the main algorithm.
It is known~\cite{ChudnovskyChudnovsky_1990,Hoeven_2001} that Fuchsian
differential equations with coefficients in~$\bQ[t]$ can be solved numerically
in quasi-linear time w.r.t. the precision.
Yet, some minor technical points must be addressed to apply the
results of the literature to our setting.
We start with reminders on the theory of linear ODEs in the complex domain
\cite[e.g.][]{Poole1936,Hille1976}.
Consider a linear differential operator
\begin{equation}\label{eq:diffop}
  P = p_m(t) \ddtp{m} + \dotsb + p_1(t) \ddt + p_0(t)
\end{equation}
of order~$m$ with coefficients in $\bQ[t]$.

Recall that $u \in \bC$ is a \emph{singular point} of~$P$ when the leading
coefficient~$p_m$ of~$P$ vanishes at~$u$.
A point that is not a singular point is called \emph{ordinary}.
Singular points are traditionally classified in two categories:
a singular point $u \in \bC$ is a \emph{regular singular} point of~$P$ if,
for $0 \leq i < m$, its multiplicity as a pole of $p_i/p_m$ is at most~$m-i$,
and an \emph{irregular singular} point otherwise.
The point at infinity in $\bP^1(\bC)$ is said to be ordinary, singular, etc.,
depending on the nature of~$0$ after
the change of variable $t \mapsto t^{-1}$.
An operator with no irregular singular point in $\bP^1(\bC)$ is called
\emph{Fuchsian}.

Fix a simply connected domain $\Omega \subseteq \bC$ containing only ordinary
points of~$P$,
and let $W$ be the space of analytic solutions
$y: \Omega \to \bC$ of the differential equation $P(y)=0$.
According to the Cauchy existence theorem for linear analytic ODEs,
$W$ is a complex vector space of dimension~$m$.
A particular solution $y \in W$ is determined by the initial values
$y(u), y'(u), \dotsc, y^{(m-1)}(u)$
at any point~$u \in \Omega$.

At a singular point, there may not be any nonzero analytic solution.
Yet, if $u$~is a \emph{regular} singular point,
the differential equation still admits $m$~linearly
independent solutions defined in the slit disk
$\left\{ u + \zeta \st |\zeta| < \eta, \zeta \notin \bR_- \right\}$
for small enough~$\eta$
and each of the form
\begin{equation} \label{eq:regsingsol}
y(u + \zeta)
  = \zeta^{\gamma} \sum_{k=0}^{\ell} y_k(\zeta) \log(\zeta)^k
  = \sum_{k=0}^{\ell} \sum_{\nu \in \gamma+\bN}^{\infty} y_{k,\nu} \zeta^{\nu}
    \log(\zeta)^k
\end{equation}
where $\gamma \in \bar\bQ$, $\ell \in \bN$, and $y_{k,\gamma} = y_k(0) \neq 0$
for exactly one~$k$ \cite[§16]{Poole1936}.
The functions~$y_k$ are analytic for~$|\zeta| < \eta$ (including at~$0$).
The algebraic numbers~$\gamma$ are called the \emph{exponents} of~$P$ at~$u$.

Suppose now that $u$~is either an ordinary point of~$P$ lying in the topological
closure~$\bar \Omega$ of~$\Omega$, or a regular singular point of~$P$ situated
on the boundary of~$\Omega$.
As a result of the previous discussion, we can choose a distinguished basis
$B_u = (\phi_{u,1}, \dots, \phi_{u,m})$ of~$W$ in which each~$\phi_{u,i}$ is
characterized by the leading monomial%
\footnote{
  More precisely, denoting $\lambda_{k,\nu}(y) = y_{k,\nu}$
  in~\eqref{eq:regsingsol}, there are~$m$ computable pairs $(\gamma_i, k_i)$
  such that, for all $i$, we have
  $\lambda_{k_i,\gamma_i}(\phi_{u,i}) = 1$,
  $\smash{\lambda_{k_j,\gamma_j}}(\phi_{u,i}) = 0$ for $j \neq i$, and
  $\lambda_{k,\nu}(\phi_{u,i}) = 0$ whenever $\nu - \gamma_i \notin \bN$.
}
$({t-u})^{\gamma} \log({t-u})^k$
of its local expansion~\eqref{eq:regsingsol} at~$u$.
At an ordinary point~$u$ for instance,
the coefficients of the decomposition of a solution~$y$
on~$B_u$ are $y^{(i)}(u)/i!$, that is, essentially the
classical initial values.
Observe that when no two exponents~$\gamma$ have the same imaginary part,
the elements of~$B_\alpha$ all have distinct asymptotic behaviours as $t \to u$.
In particular, at most one of them tends to a nonzero finite limit.
As Picard-Fuchs operators have real exponents according to
Theorem~\ref{thm:picard-fuchs}, this observation applies to them.

Let $u' \in \bar\Omega$ be a second point subject to the same restrictions
as~$u$.
Let $\Delta(u,u') \in \bC^{m \times m}$ be the transformation matrix from~$B_u$
to~$B_{u'}$.
The key to the quasi-linear complexity of our algorithm is that the entries of
this matrix can be computed efficiently, by solving the ODE with a Taylor method
in which sums of Taylor series are computed by binary splitting
\cite[item 178]{HAKMEM}, \cite{ChudnovskyChudnovsky_1990}.
The exact result we require is due to
van der Hoeven~\cite[Theorems~2.4 and 4.1]{Hoeven_2001};
see also \cite{Mezzarobba2011} for a detailed algorithm and some further
refinements.
Denote by $\mop{M}(n)$ the complexity of $n$-bit integer multiplication.

\begin{theorem}[\cite{Hoeven_2001}] \label{thm:transition matrix}
  For a fixed operator~$P$ and fixed algebraic numbers~$u,u'$ as above,
  one can compute the matrix $\Delta(u,u')$ with an entry-wise error bounded
  by~$2^{-p}$ in $O(\mop{M}(p (\log p)^2))$ operations.
\end{theorem}

Since $P$ is linear, this result suffices to implement the procedure
\textsc{DSolve} required by the main algorithm.
More precisely, suppose that $\textsc{TransitionMatrix}(P, u, u', p)$ returns a
matrix of complex intervals of width~$O(2^{-p})$ that encloses $\Delta(u,u')$
entry-wise.

\begin{definition}
  A system of \emph{good initial conditions} for~$P$ on~$\Omega$,
  denoted $[\lambda_j(y) = s_j]_{i=0}^{m'}$,
  is a finite family  of pairs $(\lambda_j, s_j)$ where $s_j \in \bC$ and
  $\lambda_j$ is a linear
  form that belongs to the dual basis of~$B_u$ for some algebraic
  point~$u \in \bar \Omega$ (which may depend on~$j$),
  with the property that $\lambda_1, \dots, \lambda_{m'}$ span the dual
  space of~$W$.

  A system of good initial conditions on $(\alpha, \beta) \subset \bR$ is a
  system of good initial conditions on
  $(\alpha, \beta) + i \, (0, \epsilon)$ for some $\epsilon > 0$.
\end{definition}

In other words, a system of good initial conditions is a choice of
coefficients of local decompositions of a solution of~$P$ whose values determine
at most one solution, and of prescribed values for these coefficients.
When the system is compatible, we say that it \emph{defines} the unique solution
of~$P$ that satisfies all the constraints.
Let us note in passing the following fact, which was used in
§\ref{sec:general-semi-algebraic}.

\begin{lemma} \label{lemma:ini-integral}
Let $u_1, \dots, u_{m'}$ be ordinary points of~$P$ such that
$\mathscr I = [ y(u_i) = s_i ]_i$
is a system of good initial conditions for $P$ on $\Omega$,
and let $u_0 \in \bar \Omega$.
Then
$\mathscr I' = [ y(u_0) = s_0, y'(u_1) = s_1, \dots, y'(u_{m'}) = s_{m'} ]$
is a system of good initial conditions for~$P\ddt$ on~$\Omega$.
\end{lemma}

\begin{proof}
The derivative $y \mapsto y'$ maps the solution space of $P \ddt$ to that
of~$P$, and its kernel consists exactly of the constant functions.
By assumption, a solution of~$P$ is completely defined by its values at
$u_1, \dots, u_m$, hence a solution of~$P \ddt$ is characterized by the
values of its derivative at the same points, along with its limit at~$u_0$.
Because $P \ddt$ has order at least~$2$ (otherwise, $\mathscr I$ would not be
a system of good initial conditions), the conditions~$y'(u_i)$ are of the
form $\lambda(y)=s$ with $\lambda$ belonging to the dual basis of some~$B_u$, as
required.
So is the condition $y(u_0)=s_0$ since $P \ddt$ has solutions with a nonzero
finite limit at~$u_0$.
\end{proof}

\begin{algorithm}[t]
    \caption{Solution of $P(y)=0$}
    \label{algo:dsolve}
    \begin{algorithmic}[1] %
      \Procedure{DSolve}{$P,
          [ \phi^\ast_{u_i, j_i}(y) = \tilde s_i ]_{i=1}^{m},
          \alpha, p$}
      \State
      \Comment $\phi^\ast_{u, j}$ is the linear form dual to the element $\phi_{u,j}$ of $B_u$
      \If{$B_\alpha$ has an element of leading monomial~$1$}
        \State{$j_0 \gets \text{its index}$}
      \Else\ 
        \textbf{return} $0$
      \EndIf
      \State $u_0 \gets \alpha$; $\tilde \Delta_0 \gets \mop{id}$
      \For{$1\leq i \leq m$}
        \Comment using interval arithmetic
        \State \label{step:transition matrix}
               $\tilde \Delta_{i} \gets \textsc{TransitionMatrix}(P, u_{i-1}, u_i, p)
               \cdot \tilde \Delta_{i-1}$
        \State $\tilde \Lambda_i \gets \text{$j_i$th row of $\Delta_i$}$
      \EndFor
      \State \label{step:solve}
             solve the linear system $\tilde \Lambda_i \cdot \tilde c = \tilde s_i$,
             $1 \leq i \leq m$ (or fail)
      \State \textbf{return} the real part of $\tilde c_{j_0}$
      \EndProcedure
    \end{algorithmic}
\end{algorithm}

Algorithm~\ref{algo:dsolve} evaluates the solution of an operator~$P$
given by a system of good initial conditions.
Note that the algorithm is allowed to fail.
It fails if the intervals $\tilde \Lambda_i$ are not accurate enough for the
linear algebra step on line~\ref{step:solve} to succeed,
or if the linear system, which is in general over-determined, has no solution.
The following proposition assumes a large enough working precision~$p$ to
ensure that this does not happen.
Additionally, we only require that the output be accurate to within~$O(2^{-p})$,
so as to absorb any loss of precision resulting from numerical stability issues
or from the use of interval arithmetic.

\begin{proposition} \label{prop:dsolve}
  Suppose that the operator~$P$ is Fuchsian with real exponents.
  Let $\alpha < \beta$ be real algebraic numbers, and let $y$ be a real analytic
  solution of $P(y)=0$ on the interval $(\alpha, \beta)$ such that $y(t)$ tends
  to a finite limit as $t \to \alpha$.
  Let~$\mathscr I = [ \lambda_i(y) = s_i ]$ be a system of good initial
  conditions for~$P$ on $(\alpha, \beta)$ that defines~$y$.

  Given the operator~$P$, the point~$\alpha$, a \emph{large enough} working
  precision~$p \in \bN$, and an approximation
  $\tilde{\mathscr I} = [ \lambda_i(y) = \tilde s_i ]$
  of~$\mathscr I$
  where $\tilde s_i$ is an interval of width at most $2^{-p}$
  containing~$s_i$,
  $\textsc{DSolve}(P, \tilde{\mathscr I}, \alpha, p)$
  (Algorithm~\ref{algo:dsolve})
  computes a real interval of width $O(2^{-p})$ containing
  $\lim_{t \to \alpha} y(t)$ in time
  $O(\mop{M}(p (\log p)^2))$.
\end{proposition}

\begin{proof}
  At the end of the loop, we have $\Delta(\alpha, u_i) \in \tilde \Delta_i$,
  $1 \leq i \leq m$, and the entries of $\tilde \Delta_i$ are intervals of
  width~$O(2^{-p})$.
  The coefficients $c = (c_i)_i$ of the decomposition of~$y$ in the
  basis~$B_\alpha$ satisfy $\Lambda_i \cdot c = s_i$ for all~$i$,
  where $\Lambda_i$ is the $j_i$th row of $\Delta(\alpha, u_i)$.
  As~$\mathscr I$ is a system of good initial conditions, the linear system
  $(\Lambda_i \cdot x = s_i)_i$ has no other solution.
  Step~\ref{step:solve} hence succeeds in solving the interval version as soon
  as the~$\tilde s_i$ and the entries of the~$\tilde \Delta_i$ are thin enough
  intervals.
  It then returns intervals of width~$O(2^{-p})$.

  We assumed that $y$ tends to a finite limit at~$\alpha$.
  It follows that the decomposition of~$y$ on~$B_\alpha$ only involves the
  basis elements with a finite limit at~$\alpha$.
  Either~$B_\alpha$ contains an element $\phi_{\alpha, j_0}$ that tends to~$1$,
  in which case $\lim_{\alpha} y = c_{j_0}$, or every
  solution that converges tends to zero, and then the limit is zero.
  Since, by assumption, $\lim_{\alpha} y$ is real, we can ignore the imaginary
  part of the computed value.
  In both cases, the algorithm, when it succeeds, returns a
  real interval of width $O(2^{-p})$ containing $\lim_{\alpha} y$.

  As for the complexity analysis,
  all $u_i$ including~$\alpha$ are algebraic, hence
  Theorem~\ref{thm:transition matrix} applies and shows that each call to
  \textsc{Transition\allowbreak{}Matrix} runs in time $O(\mop{M}(p (\log p)^2))$.
  The matrix multiplications at step~\ref{step:transition matrix} take
  $O(\mop M(p))$ operations.
  The cost of solving the linear system (which is of bounded size) is
  $O(\mop{M}(p))$ as well.
  The cost of the remaining steps is independent of~$p$.
\end{proof}

It remains to show how to implement \textsc{PickGoodPoints}.
Choosing the points at random works with probability one.
The procedure described below has the advantage of being deterministic
and implying (at least in principle) bounds on the bit size of the~$u_i$.

\begin{lemma} \label{lem:evaluation-points}
  Given $P$ and two real numbers $\alpha < \beta$, one can deterministically
  select $m$~points
  $u_1, \dots, u_m \in (\alpha, \beta) \cap \bQ$
  such that the evaluations $y \mapsto y(u_i)$ are good initial conditions
  for~$P$ on $(\alpha, \beta)$.
\end{lemma}

\begin{proof}
A sufficient condition for $y \mapsto y(u_i)$ to be good initial conditions is that
the matrix $M = (\psi_j(u_i))_{i,j}$, for some basis~$(\psi_j)$ of~$W$, be
invertible.
Let $K \subset (\alpha, \beta)$ be a closed interval with rational endpoints
containing only ordinary points.
Let $u_1 = \min K$.
Assume without loss of generality $u_1 = 0$, and take~$(\psi_j) = B_{u_1}$.
The matrix~$M$ is then of the form $(u_i^{j-1} + \eta_j(u_i))_{i,j=1}^m$
where, for all~$j$, $\eta_j(u) = O(u^m)$ as $u \to 0$.
In fact, there exists a computable~\cite[e.g.,][]{Hoeven_2001} constant~$C$
such that $|\eta_j(u)| \leq C |u|^m$ for all $u \in K$.
Therefore, one can compute a value $\epsilon > 0$ such that~$M$ is invertible
for any distinct $u_2, \dots, u_m$ in $(0, \epsilon)$.
The result follows.
\end{proof}

In practice, one can reduce the number of recursive calls in the main algorithm
by replacing, when possible, some of the conditions $y(u_i)=s_i$ by conditions
that result from the continuity of $v(t)$ at exceptional
points, or from its analyticity at singular points of the Picard-Fuchs operator
lying in~$\bR \setminus \Sigma$.
For instance, a solution that is analytic at~$u$ must lie in the subspace
spanned by the elements of~$B_u$ of leading term $(z-u)^\gamma$ with
$\gamma \in \bN$ and no logarithmic part.

\section{Complexity analysis}\label{sec:complexity}

Let us finally study the complexity of Algorithm~\ref{algo:volume} to
conclude the proof of Theorem~\ref{thm:main}.
For fixed $(f_1, \dots, f_r)$, all intermediate data (Picard-Fuchs
equations, critical values and specialization points chosen for the recursive
calls) are fixed thanks to the deterministic behaviour of \textsc{PickGoodPoints}
(Lemma~\ref{lem:evaluation-points}).
Thus, the number of recursive calls does not depend on~$p$.

Now, the main point is to observe that, by Proposition~\ref{prop:spec-dsolve},
performing recursive calls with precision $p+O(1)$ is enough.
One can make the width of the output interval smaller than~$2^{-p}$ by doubling~$p$ and
re-running the algorithm (if necessary with a more accurate approximation
  of~$\mathscr I$) a bounded number of times.
By Proposition~\ref{prop:dsolve}, the total cost of the calls to \textsc{DSolve}
is $O(\mop M(p \log(p)^2))$.
The only other step whose complexity depends on~$p$ is the computation of real roots
of fixed univariate polynomials in the base case, which takes $O(M(p))$ operations using Newton's method.
Using the bound $M(p) = O(p \log(p)^{1+\epsilon})$, Theorem~\ref{thm:main} follows.

This theorem ignores the dependency of the cost on the
dimension~$n$ of the ambient space or the maximum degree~$D$ of the input
polynomials.
Under some assumptions, one can bound the number of recursive calls arithmetic
cost of computing Picard-Fuchs equations and critical values as follows.
First consider Algorithm~\ref{algo:volumerec}, and let $\delta$ be the
degree of $f$. By Lemma~\ref{lemma:sard-algo}, the number of critical values and
the cost of computing them are bounded by $\delta^{O(n)}$;
in the notation of the algorithm, this shows that $\ell \leq \delta^{O(n)}$.

Under Dimca's conjecture~\cite{Dimca_1991}, the cost of computing the
Picard-Fuchs equation is $\delta^{O(n)}$ and it has order $m \leq \delta^n$
according to the discussion following Theorem~\ref{thm:algo-periods}.
One can likely obtain the same bounds without this conjecture by replacing
the deformed equation of Section~\ref{sec:general-semi-algebraic} by
$f - t \sum_i x_i^{\delta+n}$, which permits using the ``regular case''
of~\cite{BostanLairezSalvy_2013}.
Solving the recurrence $C(n+1, \delta) = \delta^{O(n)} C(n,\delta)$
shows that the algebraic steps of Algorithm~\ref{algo:volumerec} take
$\delta^{O(n^2)}$ operations in~$\bQ$.

Turning to Algorithm~\ref{algo:volume}, Lemma~\ref{lemma:sard-algo} and
Theorem~\ref{thm:algo-periods} show that the cost
of the calls to \textsc{CriticalValues} and \textsc{PicardFuchs} are dominated
by that of the calls to Algorithm~\ref{algo:volumerec} (with an input
polynomial of degree $\delta\leq rD$).
Therefore, the algebraic steps use $(rD)^{O(n^2)}$
operations in $\bQ$ in total, as announced in §\ref{sec:intro}.

We leave for future research the question of analyzing the boolean
cost of the full algorithm with respect to $n$, $D$, and the bit size of the input
coefficients.
This requires significantly more work,
as one first needs to control the bit size of the points picked by
\textsc{PickGoodPoints} in the recursive calls.
Additionally, to the best of our knowledge, no analogue of
Theorem~\ref{thm:transition matrix} fully taking into account the order, degree,
and coefficient size of the operator~$P$ is available in the literature.

\section{Conclusion}

Our algorithm generalizes to non-basic bounded semi-algebraic sets since their volume 
can be written as a linear combination with $\pm 1$ coefficients of volumes of
\emph{basic} semi-algebraic sets.

An important question that we leave for future work is that of the practicality of
our approach.
While the worst-case complexity bound is exponential in~$n^2$,
there are a number of opportunities to exploit special features of the input
that could help handling nontrivial examples in practice.
In particular:
\begin{enumerate*}[(1)]
\item the number of recursive calls only depends on the number
of \emph{real} critical points;
\item as already noted, it can be reduced by exploiting some knowledge
of the continuity of the slice volume function or its analyticity at exceptional
points;
\item it turns out that, in our case, the integral
appearing in Theorem~\ref{thm:volume-section2} always is singular at infinity,
and, as a consequence, the Picard-Fuchs equations we encounter do
not reach the worst-case degree bounds.
\end{enumerate*}
Ideally, one may hope to refine the complexity analysis to reflect some of
these observations. %

Another natural question is to extend the algorithm to unbounded semi-algebraic
sets of finite volume, or even real periods in general, using the ideas
in~\cite{Viu-Sos_2015}.
Note also that, using quantifier elimination \cite[e.g.,][]{BPR96}, boundedness can be verified in boolean time $q(rD)^{O(n)}$
where $q$ bounds the bit size of the input coefficients. %

Finally, it is plausible that an algorithm of a similar structure
but using numerical quadrature recursively instead of solving Picard-Fuchs equations
would also have polynomial complexity in the precision for fixed~$n$ and be faster
at medium precision.



\begin{thebibliography}{00}


\ifx \showCODEN    \undefined \def \showCODEN     #1{\unskip}     \fi
\ifx \showDOI      \undefined \def \showDOI       #1{#1}\fi
\ifx \showISBNx    \undefined \def \showISBNx     #1{\unskip}     \fi
\ifx \showISBNxiii \undefined \def \showISBNxiii  #1{\unskip}     \fi
\ifx \showISSN     \undefined \def \showISSN      #1{\unskip}     \fi
\ifx \showLCCN     \undefined \def \showLCCN      #1{\unskip}     \fi
\ifx \shownote     \undefined \def \shownote      #1{#1}          \fi
\ifx \showarticletitle \undefined \def \showarticletitle #1{#1}   \fi
\ifx \showURL      \undefined \def \showURL       {\relax}        \fi
\providecommand\bibfield[2]{#2}
\providecommand\bibinfo[2]{#2}
\providecommand\natexlab[1]{#1}
\providecommand\showeprint[2][]{arXiv:#2}

\bibitem[\protect\citeauthoryear{Bailey and Borwein}{Bailey and
  Borwein}{2011}]%
        {BaileyBorwein_2011}
\bibfield{author}{\bibinfo{person}{D. Bailey} {and} \bibinfo{person}{J.~M.
  Borwein}.} \bibinfo{year}{2011}\natexlab{}.
\newblock \showarticletitle{High-precision numerical integration: progress and
  challenges}.
\newblock \bibinfo{journal}{{\em Journal of Symbolic Computation\/}}
  \bibinfo{volume}{46}, \bibinfo{number}{7} (\bibinfo{year}{2011}),
  \bibinfo{pages}{741--754}.
\newblock
\showISSN{0747-7171}



\bibitem[\protect\citeauthoryear{Basu, Pollack, and Roy}{Basu
  et~al\mbox{.}}{1996}]%
        {BPR96}
\bibfield{author}{\bibinfo{person}{S. Basu}, \bibinfo{person}{R. Pollack},
  {and} \bibinfo{person}{M.-F. Roy}.} \bibinfo{year}{1996}\natexlab{}.
\newblock \showarticletitle{On the combinatorial and algebraic complexity of
  quantifier elimination}.
\newblock \bibinfo{journal}{{\it J. ACM}} \bibinfo{volume}{43},
  \bibinfo{number}{6} (\bibinfo{year}{1996}), \bibinfo{pages}{1002--1045}.
\newblock


\bibitem[\protect\citeauthoryear{Basu, Pollack, and Roy}{Basu
  et~al\mbox{.}}{2006}]%
        {BaPoRo06}
\bibfield{author}{\bibinfo{person}{S. Basu}, \bibinfo{person}{R. Pollack},
  {and} \bibinfo{person}{M.-F. Roy}.} \bibinfo{year}{2006}\natexlab{}.
\newblock \bibinfo{booktitle}{{\em Algorithms in real algebraic geometry\/}
  (\bibinfo{edition}{second} ed.)}. \bibinfo{series}{Algorithms and Computation
  in Mathematics}, Vol.~\bibinfo{volume}{10}.
\newblock \bibinfo{publisher}{Springer}.
\newblock


\bibitem[\protect\citeauthoryear{Beeler, Gosper, and Schroeppel}{Beeler
  et~al\mbox{.}}{1972}]%
        {HAKMEM}
\bibfield{author}{\bibinfo{person}{M. Beeler}, \bibinfo{person}{R.~W. Gosper},
  {and} \bibinfo{person}{R. Schroeppel}.} \bibinfo{year}{1972}\natexlab{}.
\newblock \bibinfo{booktitle}{{\em Hakmem}}.
\newblock \bibinfo{type}{AI Memo} 239. \bibinfo{institution}{MIT Artificial
  Intelligence Laboratory}.
\newblock



\bibitem[\protect\citeauthoryear{Bostan, Lairez, and Salvy}{Bostan
  et~al\mbox{.}}{2013}]%
        {BostanLairezSalvy_2013}
\bibfield{author}{\bibinfo{person}{A. Bostan}, \bibinfo{person}{P. Lairez},
  {and} \bibinfo{person}{B. Salvy}.} \bibinfo{year}{2013}\natexlab{}.
\newblock \showarticletitle{Creative telescoping for rational functions using
  the Griffiths–Dwork method}. In \bibinfo{booktitle}{{\em ISSAC 2013}}.
  \bibinfo{publisher}{{ACM}}, \bibinfo{pages}{93--100}.
\newblock
\showISBNx{978-1-4503-2059-7}



\bibitem[\protect\citeauthoryear{B{\"u}rgisser, Cucker, and
  Lairez}{B{\"u}rgisser et~al\mbox{.}}{2019}]%
        {BCL18}
\bibfield{author}{\bibinfo{person}{P. B{\"u}rgisser}, \bibinfo{person}{F.
  Cucker}, {and} \bibinfo{person}{P. Lairez}.} \bibinfo{year}{2019}\natexlab{}.
\newblock \showarticletitle{Computing the homology of basic semialgebraic sets
  in weak exponential time}.
\newblock \bibinfo{journal}{{\it J. ACM}} \bibinfo{volume}{66},
  \bibinfo{number}{1} (\bibinfo{year}{2019}), \bibinfo{pages}{5}.
\newblock


\bibitem[\protect\citeauthoryear{Canny}{Canny}{1988}]%
        {Canny88}
\bibfield{author}{\bibinfo{person}{J. Canny}.} \bibinfo{year}{1988}\natexlab{}.
\newblock \bibinfo{booktitle}{{\em The complexity of robot motion planning}}.
\newblock \bibinfo{publisher}{MIT press}.
\newblock


\bibitem[\protect\citeauthoryear{Christol}{Christol}{1985}]%
        {Christol_1985}
\bibfield{author}{\bibinfo{person}{G. Christol}.}
  \bibinfo{year}{1985}\natexlab{}.
\newblock \showarticletitle{Diagonales de fractions rationnelles et équations
  de Picard-Fuchs}.
\newblock \bibinfo{journal}{{\em Groupe de travail d'analyse ultramétrique\/}}
  \bibinfo{volume}{12}, \bibinfo{number}{1} (\bibinfo{year}{1985}).
\newblock


\bibitem[\protect\citeauthoryear{Chudnovsky and Chudnovsky}{Chudnovsky and
  Chudnovsky}{1990}]%
        {ChudnovskyChudnovsky_1990}
\bibfield{author}{\bibinfo{person}{D.~V. Chudnovsky} {and}
  \bibinfo{person}{G.~V. Chudnovsky}.} \bibinfo{year}{1990}\natexlab{}.
\newblock \showarticletitle{Computer algebra in the service of mathematical
  physics and number theory}.
\newblock In \bibinfo{booktitle}{{\em Computers in mathematics (Stanford, CA,
  1986)}}. \bibinfo{series}{Lecture Notes in Pure and Appl. Math.},
  Vol.~\bibinfo{volume}{125}. \bibinfo{publisher}{Dekker},
  \bibinfo{pages}{109--232}.
\newblock


\bibitem[\protect\citeauthoryear{Chyzak}{Chyzak}{2000}]%
        {Chyzak_2000}
\bibfield{author}{\bibinfo{person}{F. Chyzak}.}
  \bibinfo{year}{2000}\natexlab{}.
\newblock \showarticletitle{An extension of Zeilberger's fast algorithm to
  general holonomic functions}.
\newblock \bibinfo{journal}{{\em Discrete Mathematics\/}}
  \bibinfo{volume}{217}, \bibinfo{number}{1-3} (\bibinfo{year}{2000}),
  \bibinfo{pages}{115--134}.
\newblock
\showISSN{0012-365X}



\bibitem[\protect\citeauthoryear{Collins}{Collins}{1975}]%
        {Collins}
\bibfield{author}{\bibinfo{person}{G.~E. Collins}.}
  \bibinfo{year}{1975}\natexlab{}.
\newblock \showarticletitle{Quantifier elimination for real closed fields by
  cylindrical algebraic decompostion}. In \bibinfo{booktitle}{{\em Automata
  Theory and Formal Languages 2nd GI Conference Kaiserslautern, May 20--23,
  1975}}. Springer, \bibinfo{pages}{134--183}.
\newblock


\bibitem[\protect\citeauthoryear{Dimca}{Dimca}{1991}]%
        {Dimca_1991}
\bibfield{author}{\bibinfo{person}{A. Dimca}.} \bibinfo{year}{1991}\natexlab{}.
\newblock \showarticletitle{On the de Rham cohomology of a hypersurface
  complement}.
\newblock \bibinfo{journal}{{\em American Journal of Mathematics\/}}
  \bibinfo{volume}{113}, \bibinfo{number}{4} (\bibinfo{year}{1991}),
  \bibinfo{pages}{763--771}.
\newblock
\showISSN{0002-9327}



\bibitem[\protect\citeauthoryear{Dyer and Frieze}{Dyer and Frieze}{1988}]%
        {DyerFrieze_1988}
\bibfield{author}{\bibinfo{person}{M. Dyer} {and} \bibinfo{person}{A. Frieze}.}
  \bibinfo{year}{1988}\natexlab{}.
\newblock \showarticletitle{On the complexity of computing the volume of a
  polyhedron}.
\newblock \bibinfo{journal}{{\it SIAM J. Comput.}} \bibinfo{volume}{17},
  \bibinfo{number}{5} (\bibinfo{year}{1988}), \bibinfo{pages}{967--974}.
\newblock
\showISSN{0097-5397}



\bibitem[\protect\citeauthoryear{Dyer, Frieze, and Kannan}{Dyer
  et~al\mbox{.}}{1991}]%
        {DyerFriezeKannan_1991}
\bibfield{author}{\bibinfo{person}{M. Dyer}, \bibinfo{person}{A. Frieze}, {and}
  \bibinfo{person}{R. Kannan}.} \bibinfo{year}{1991}\natexlab{}.
\newblock \showarticletitle{A random polynomial-time algorithm for
  approximating the volume of convex bodies}.
\newblock \bibinfo{journal}{{\it J. ACM}} \bibinfo{volume}{38},
  \bibinfo{number}{1} (\bibinfo{year}{1991}), \bibinfo{pages}{1--17}.
\newblock
\showISSN{0004-5411}



\bibitem[\protect\citeauthoryear{Grothendieck}{Grothendieck}{1966}]%
        {Grothendieck_1966a}
\bibfield{author}{\bibinfo{person}{A. Grothendieck}.}
  \bibinfo{year}{1966}\natexlab{}.
\newblock \showarticletitle{On the de Rham cohomology of algebraic varieties}.
\newblock \bibinfo{journal}{{\em Institut des Hautes \'Etudes Scientifiques.
  Publications Mathématiques\/}} \bibinfo{number}{29} (\bibinfo{year}{1966}),
  \bibinfo{pages}{95--103}.
\newblock
\showISSN{0073-8301}


\bibitem[\protect\citeauthoryear{Henrion, Lasserre, and Savorgnan}{Henrion
  et~al\mbox{.}}{2009}]%
        {HLS09}
\bibfield{author}{\bibinfo{person}{D. Henrion}, \bibinfo{person}{J.-B.
  Lasserre}, {and} \bibinfo{person}{C. Savorgnan}.}
  \bibinfo{year}{2009}\natexlab{}.
\newblock \showarticletitle{Approximate volume and integration for basic
  semialgebraic sets}.
\newblock \bibinfo{journal}{{\em SIAM review\/}} \bibinfo{volume}{51},
  \bibinfo{number}{4} (\bibinfo{year}{2009}), \bibinfo{pages}{722--743}.
\newblock


\bibitem[\protect\citeauthoryear{Hille}{Hille}{1976}]%
        {Hille1976}
\bibfield{author}{\bibinfo{person}{E. Hille}.} \bibinfo{year}{1976}\natexlab{}.
\newblock \bibinfo{booktitle}{{\em Ordinary differential equations in the
  complex domain}}.
\newblock \bibinfo{publisher}{Wiley}.
\newblock
\newblock
\shownote{{D}over reprint, 1997.}


\bibitem[\protect\citeauthoryear{Katz}{Katz}{1971}]%
        {Katz_1971}
\bibfield{author}{\bibinfo{person}{N.~M. Katz}.}
  \bibinfo{year}{1971}\natexlab{}.
\newblock \showarticletitle{The regularity theorem in algebraic geometry}.
\newblock In \bibinfo{booktitle}{{\em Actes du Congrès International des
  Mathématiciens (Nice, 1970), Tome 1}}.
  \bibinfo{publisher}{{Gauthier-Villars}}, \bibinfo{pages}{437--443}.
\newblock


\bibitem[\protect\citeauthoryear{Khachiyan}{Khachiyan}{1989}]%
        {Kh89}
\bibfield{author}{\bibinfo{person}{L. Khachiyan}.}
  \bibinfo{year}{1989}\natexlab{}.
\newblock \showarticletitle{The problem of computing the volume of polytopes is
  NP-hard}.
\newblock \bibinfo{journal}{{\em Uspekhi Mat. Nauk\/}} \bibinfo{volume}{44},
  \bibinfo{number}{3} (\bibinfo{year}{1989}), \bibinfo{pages}{199--200}.
\newblock


\bibitem[\protect\citeauthoryear{Khachiyan}{Khachiyan}{1993}]%
        {Kh93}
\bibfield{author}{\bibinfo{person}{L. Khachiyan}.}
  \bibinfo{year}{1993}\natexlab{}.
\newblock \showarticletitle{Complexity of polytope volume computation}.
\newblock In \bibinfo{booktitle}{{\em New trends in discrete and computational
  geometry}}. \bibinfo{publisher}{Springer}, \bibinfo{pages}{91--101}.
\newblock


\bibitem[\protect\citeauthoryear{Khachiyan and Porkolab}{Khachiyan and
  Porkolab}{2000}]%
        {KhachiyanPorkolab2000}
\bibfield{author}{\bibinfo{person}{L. Khachiyan} {and} \bibinfo{person}{L.
  Porkolab}.} \bibinfo{year}{2000}\natexlab{}.
\newblock \showarticletitle{Integer Optimization on Convex Semialgebraic Sets}.
\newblock \bibinfo{journal}{{\em Discrete \& Computational Geometry\/}}
  \bibinfo{volume}{23}, \bibinfo{number}{2} (\bibinfo{date}{Feb.}
  \bibinfo{year}{2000}), \bibinfo{pages}{207--224}.
\newblock
\showISSN{1432-0444}



\bibitem[\protect\citeauthoryear{Koiran}{Koiran}{1995}]%
        {Ko95}
\bibfield{author}{\bibinfo{person}{P. Koiran}.}
  \bibinfo{year}{1995}\natexlab{}.
\newblock \showarticletitle{Approximating the volume of definable sets}. In
  \bibinfo{booktitle}{{\em Foundations of Computer Science}}. IEEE,
  \bibinfo{pages}{134--141}.
\newblock


\bibitem[\protect\citeauthoryear{Kontsevich and Zagier}{Kontsevich and
  Zagier}{2001}]%
        {KontsevichZagier_2001}
\bibfield{author}{\bibinfo{person}{M. Kontsevich} {and} \bibinfo{person}{D.
  Zagier}.} \bibinfo{year}{2001}\natexlab{}.
\newblock \showarticletitle{Periods}.
\newblock In \bibinfo{booktitle}{{\em Mathematics unlimited}}.
  \bibinfo{publisher}{{Springer}}, \bibinfo{pages}{771--808}.
\newblock


\bibitem[\protect\citeauthoryear{Korda and Henrion}{Korda and Henrion}{2018}]%
        {KordaHenrion_2017}
\bibfield{author}{\bibinfo{person}{M. Korda} {and} \bibinfo{person}{D.
  Henrion}.} \bibinfo{year}{2018}\natexlab{}.
\newblock \showarticletitle{Convergence rates of moment-sum-of-squares
  hierarchies for volume approximation of semialgebraic sets}.
\newblock \bibinfo{journal}{{\em Optimization Letters\/}} \bibinfo{volume}{12},
  \bibinfo{number}{3} (\bibinfo{year}{2018}), \bibinfo{pages}{435--442}.
\newblock
\showISSN{1862-4472, 1862-4480}



\bibitem[\protect\citeauthoryear{Koutschan}{Koutschan}{2010}]%
        {Koutschan_2010}
\bibfield{author}{\bibinfo{person}{C. Koutschan}.}
  \bibinfo{year}{2010}\natexlab{}.
\newblock \showarticletitle{A fast approach to creative telescoping}.
\newblock \bibinfo{journal}{{\em Mathematics in Computer Science\/}}
  \bibinfo{volume}{4}, \bibinfo{number}{2-3} (\bibinfo{year}{2010}),
  \bibinfo{pages}{259--266}.
\newblock



\bibitem[\protect\citeauthoryear{Lairez}{Lairez}{2016}]%
        {Lairez_2016}
\bibfield{author}{\bibinfo{person}{P. Lairez}.}
  \bibinfo{year}{2016}\natexlab{}.
\newblock \showarticletitle{Computing periods of rational integrals}.
\newblock \bibinfo{journal}{{\it Math. Comp.}} \bibinfo{volume}{85},
  \bibinfo{number}{300} (\bibinfo{year}{2016}), \bibinfo{pages}{1719--1752}.
\newblock



\bibitem[\protect\citeauthoryear{Leray}{Leray}{1959}]%
        {Leray_1959}
\bibfield{author}{\bibinfo{person}{J. Leray}.} \bibinfo{year}{1959}\natexlab{}.
\newblock \showarticletitle{Le calcul différentiel et intégral sur une
  variété analytique complexe (Problème de Cauchy, III)}.
\newblock \bibinfo{journal}{{\em Bulletin de la Société mathématique de
  France\/}}  \bibinfo{volume}{87} (\bibinfo{year}{1959}),
  \bibinfo{pages}{81--180}.
\newblock
\showISSN{0037-9484, 2102-622X}



\bibitem[\protect\citeauthoryear{Lipshitz}{Lipshitz}{1988}]%
        {Lipshitz_1988}
\bibfield{author}{\bibinfo{person}{L. Lipshitz}.}
  \bibinfo{year}{1988}\natexlab{}.
\newblock \showarticletitle{The diagonal of a D-finite power series is
  D-finite}.
\newblock \bibinfo{journal}{{\em Journal of Algebra\/}} \bibinfo{volume}{113},
  \bibinfo{number}{2} (\bibinfo{year}{1988}), \bibinfo{pages}{373--378}.
\newblock
\showISSN{0021-8693}



\bibitem[\protect\citeauthoryear{Mezzarobba}{Mezzarobba}{2011}]%
        {Mezzarobba2011}
\bibfield{author}{\bibinfo{person}{M. Mezzarobba}.}
  \bibinfo{year}{2011}\natexlab{}.
\newblock {\em \bibinfo{title}{Autour de l'{\'e}valuation num{\'e}rique des
  fonctions {D}-finies}}.
\newblock Th{\`e}se de doctorat. \bibinfo{school}{{\'E}cole polytechnique}.
\newblock



\bibitem[\protect\citeauthoryear{Mezzarobba}{Mezzarobba}{2016}]%
        {Mezzarobba_2016}
\bibfield{author}{\bibinfo{person}{M. Mezzarobba}.}
  \bibinfo{year}{2016}\natexlab{}.
\newblock \bibinfo{title}{Rigorous multiple-precision evaluation of {D}-finite
  functions in {SageMath}}.
\newblock   (\bibinfo{year}{2016}).
\newblock
\showeprint[arxiv]{1607.01967}


\bibitem[\protect\citeauthoryear{Monsky}{Monsky}{1972}]%
        {Monsky_1972}
\bibfield{author}{\bibinfo{person}{P. Monsky}.}
  \bibinfo{year}{1972}\natexlab{}.
\newblock \showarticletitle{Finiteness of de Rham cohomology}.
\newblock \bibinfo{journal}{{\em American Journal of Mathematics\/}}
  \bibinfo{volume}{94} (\bibinfo{year}{1972}), \bibinfo{pages}{237--245}.
\newblock
\showISSN{0002-9327}


\bibitem[\protect\citeauthoryear{Oaku}{Oaku}{2013}]%
        {Oaku_2013}
\bibfield{author}{\bibinfo{person}{T. Oaku}.} \bibinfo{year}{2013}\natexlab{}.
\newblock \showarticletitle{Algorithms for integrals of holonomic functions
  over domains defined by polynomial inequalities}.
\newblock \bibinfo{journal}{{\em Journal of Symbolic Computation\/}}
  \bibinfo{volume}{50} (\bibinfo{year}{2013}), \bibinfo{pages}{1--27}.
\newblock
\showISSN{0747-7171}



\bibitem[\protect\citeauthoryear{Pham}{Pham}{2011}]%
        {Pham_2011}
\bibfield{author}{\bibinfo{person}{F. Pham}.} \bibinfo{year}{2011}\natexlab{}.
\newblock \bibinfo{booktitle}{{\em Singularities of integrals: homology,
  hyperfunctions and microlocal analysis}}.
\newblock \bibinfo{publisher}{Springer ; EDP Sciences}.
\newblock
\showISBNx{978-0-85729-603-0 978-0-85729-602-3 978-2-7598-0363-7}


\bibitem[\protect\citeauthoryear{Picard}{Picard}{1902}]%
        {Picard_1902}
\bibfield{author}{\bibinfo{person}{E. Picard}.}
  \bibinfo{year}{1902}\natexlab{}.
\newblock \showarticletitle{Sur les périodes des intégrales doubles et sur
  une classe d'équations différentielles linéaires}. In
  \bibinfo{booktitle}{{\em Comptes rendus hebdomadaires des séances de
  l'Académie des sciences}}, Vol.~\bibinfo{volume}{134}.
  \bibinfo{publisher}{Gauthier-Villars}, \bibinfo{pages}{69--71}.
\newblock



\bibitem[\protect\citeauthoryear{Poole}{Poole}{1936}]%
        {Poole1936}
\bibfield{author}{\bibinfo{person}{E.~G.~C. Poole}.}
  \bibinfo{year}{1936}\natexlab{}.
\newblock \bibinfo{booktitle}{{\em Introduction to the theory of linear
  differential equations}}.
\newblock \bibinfo{publisher}{Clarendon Press}.
\newblock


\bibitem[\protect\citeauthoryear{Safey El~Din and Schost}{Safey El~Din and
  Schost}{2003}]%
        {SaSc03}
\bibfield{author}{\bibinfo{person}{M. Safey El~Din} {and}
  \bibinfo{person}{{\'E}. Schost}.} \bibinfo{year}{2003}\natexlab{}.
\newblock \showarticletitle{Polar varieties and computation of one point in
  each connected component of a smooth real algebraic set}. In
  \bibinfo{booktitle}{{\em ISSAC 2003}}. ACM, \bibinfo{pages}{224--231}.
\newblock


\bibitem[\protect\citeauthoryear{{Safey El Din} and Schost}{{Safey El Din} and
  Schost}{2017}]%
        {SaSc17}
\bibfield{author}{\bibinfo{person}{M. {Safey El Din}} {and}
  \bibinfo{person}{{\'E}. Schost}.} \bibinfo{year}{2017}\natexlab{}.
\newblock \showarticletitle{A nearly optimal algorithm for deciding
  connectivity queries in smooth and bounded real algebraic sets}.
\newblock \bibinfo{journal}{{\it J. ACM}} \bibinfo{volume}{63},
  \bibinfo{number}{6} (\bibinfo{year}{2017}), \bibinfo{pages}{48}.
\newblock


\bibitem[\protect\citeauthoryear{Safey El~Din and Zhi}{Safey El~Din and
  Zhi}{2010}]%
        {SafeyElDinZhi2010}
\bibfield{author}{\bibinfo{person}{M. Safey El~Din} {and} \bibinfo{person}{L.
  Zhi}.} \bibinfo{year}{2010}\natexlab{}.
\newblock \showarticletitle{Computing Rational Points in Convex Semialgebraic
  Sets and Sum of Squares Decompositions}.
\newblock \bibinfo{journal}{{\em SIAM Journal on Optimization\/}}
  \bibinfo{volume}{20}, \bibinfo{number}{6} (\bibinfo{date}{Jan.}
  \bibinfo{year}{2010}), \bibinfo{pages}{2876--2889}.
\newblock
\showISSN{1052-6234}



\bibitem[\protect\citeauthoryear{Tarski}{Tarski}{1998}]%
        {Tarski}
\bibfield{author}{\bibinfo{person}{A. Tarski}.}
  \bibinfo{year}{1998}\natexlab{}.
\newblock \showarticletitle{A decision method for elementary algebra and
  geometry}.
\newblock In \bibinfo{booktitle}{{\em Quantifier elimination and cylindrical
  algebraic decomposition}}. \bibinfo{publisher}{Springer},
  \bibinfo{pages}{24--84}.
\newblock


\bibitem[\protect\citeauthoryear{van~der Hoeven}{van~der Hoeven}{2001}]%
        {Hoeven_2001}
\bibfield{author}{\bibinfo{person}{J. van~der Hoeven}.}
  \bibinfo{year}{2001}\natexlab{}.
\newblock \showarticletitle{Fast Evaluation of Holonomic Functions Near and in
  Regular Singularities}.
\newblock \bibinfo{journal}{{\em Journal of Symbolic Computation\/}}
  \bibinfo{volume}{31}, \bibinfo{number}{6} (\bibinfo{year}{2001}),
  \bibinfo{pages}{717--743}.
\newblock
\showISSN{07477171}



\bibitem[\protect\citeauthoryear{Viu-Sos}{Viu-Sos}{9 03}]%
        {Viu-Sos_2015}
\bibfield{author}{\bibinfo{person}{J. Viu-Sos}.}
  \bibinfo{year}{2015-09-03}\natexlab{}.
\newblock \bibinfo{title}{A semi-canonical reduction for periods of
  Kontsevich-Zagier}.
\newblock   (\bibinfo{year}{2015-09-03}).
\newblock
\showeprint[arxiv]{1509.01097}


\bibitem[\protect\citeauthoryear{von~zur Gathen and Gerhard}{von~zur Gathen and
  Gerhard}{1999}]%
        {GathenGerhard_1999}
\bibfield{author}{\bibinfo{person}{J. von~zur Gathen} {and} \bibinfo{person}{J.
  Gerhard}.} \bibinfo{year}{1999}\natexlab{}.
\newblock \bibinfo{booktitle}{{\em Modern computer algebra}}.
\newblock \bibinfo{publisher}{{Cambridge University Press}}.
\newblock
\showISBNx{0-521-64176-4}


\end{thebibliography}
\end{document}